\documentclass[letterpaper,onecolumn,prx,aps,showpacs,superscriptaddress,longbibliography,floatfix]{revtex4-2}

\usepackage{amssymb, amsmath, bbm, amsthm, changes, graphicx, dsfont, verbatim, appendix, physics} 

\usepackage[hyperindex, breaklinks]{hyperref}
\hypersetup{
     colorlinks=true,        
     citecolor=blue,    
     filecolor=blue,      		
     urlcolor=blue,           	
    runcolor=cyan,
}

\usepackage{enumitem}   
\usepackage{xcolor}
\usepackage{orcidlink}

\def\Z2{\mathbb Z_2}

\def\Ht{\widetilde{H}}
\def\Hh{\widehat{H}}
\def\lm{l_{\textrm{MIN}}}
\def\dist{\textrm{dist}}
\def\supp{\textrm{support}}

\def\b|{\big \|}
\def\om{\omega}
\def\Eom{\underset{\omega}{\mathbb{E}} }

\theoremstyle{definition}
\newtheorem{Lemma}{Lemma}

\newtheorem{Theorem}[Lemma]{Theorem}

\begin{document}
\title{Dynamics of many-body localized systems: logarithmic lightcones and $\log \, t$-law of $\alpha$-Rényi entropies}

\author{Daniele Toniolo\,\orcidlink{0000-0003-2517-0770}}
\email[]{daniele.toniolo@uni-tuebingen.de}
\email[]{danielet@alumni.ntnu.no}
\affiliation{Department of Physics and Astronomy, University College London, Gower Street, London WC1E 6BT, United Kingdom}
\affiliation{Department of Mathematics, University of Tübingen, Auf der Morgenstelle 10, 72076 Tübingen, Germany}

\author{Sougato Bose\,\orcidlink{0000-0001-8726-0566}}
\affiliation{Department of Physics and Astronomy, University College London, Gower Street, London WC1E 6BT, United Kingdom}


\begin{abstract}
In the context of the Many-Body-Localization phenomenology we consider arbitrarily large one-dimensional local spin systems, the XXZ model with random magnetic field is a prototypical example. Without assuming the existence of exponentially localized integrals of motion (LIOM), but assuming instead that the system's dynamics gives rise to a Lieb-Robinson bound (L-R) with a logarithmic lightcone, we rigorously evaluate the dynamical generation, starting from a generic product state, of $ \alpha$-Rényi entropies, with $ \alpha $ close to one, obtaining a $\log \, t$-law, that denotes a slow spread of entanglement. This is in sharp contrast with Anderson localized phases that show no dynamically generated entanglement. To prove this result we apply a general theory recently developed by us in \cite{Toniolo_2024_2} that quantitatively relates the L-R bounds of a local Hamiltonian with the dynamical generation of entanglement. Assuming instead the existence of LIOM we provide new independent proofs of the known facts that the L-R bound of the system's dynamics has a logarithmic lightcone and show that the dynamical generation of the von Neumann entropy has for large times a $ \log \, t$-shape.
L-R bounds, that quantify the dynamical spreading of local operators, may be easier to measure in experiments in comparison to global quantities such as entanglement.
\end{abstract}

\maketitle


\section{Introduction}

One-dimensional interacting and disordered systems in the many-body localized (MBL) phase  have been regarded as non thermalizing \cite{Basko_2006, Gornyi_2005, Oganesyan_2007, Znidaric_2008, Pal_2010}. Review articles on MBL are for example \cite{Alet_review_2018, Abanin_2019, Sierant_review_2024}. MBL has been experimentally explored both in one \cite{Schreiber_2015,Lukin_2019} and in two dimensions \cite{Choi_2016,Bordia_2017}. Nevertheless the works \cite{De_Roeck_2017_1,Luitz_2017}, see also \cite{Morningstar_Huse_2022,Morningstar_Huse_2023}, started a long term debate about the stability of MBL with respect to the existence of rare ergodic regions. An experimental signature of the so called avalanche scenario has been discussed in \cite{Leonard_2023}.
Other references relevant for the discussion about the stability and the actual existence of genuine MBL phases are \cite{Suntajs_2020,Abanin_2021,Sierant_2020,Panda_2019,Kiefer_2020}.

Recently we have provided in \cite{Toniolo_2024_1} an analysis of the stability of the dynamics of systems with logarithmic lightcones that are relevant to MBL, proving that a local perturbation on top of a system with a slow dynamics (namely a Lieb-Robinson bound with a logarithmic lightcone) to propagate across a region of size $ l $ takes a time proportional to the exponential of such a distance. Numerical works \cite{Sierant_2024,Scocco_2024} studying, among other things, the dynamical effects of joining together a one-dimensional MBL system with an ergodic one or a thermal bath agree with our findings in \cite{Toniolo_2024_1}.

 One of the most prominent distinguishing features of MBL from Anderson localization is the different dynamics of entanglement as measured by quantum entropies. It has been shown by numerical simulation of (small) systems' dynamics, for example in \cite{Znidaric_2008,Bardarson_2012,Serbyn_2013,Huang_2021}, and many other works, and even experimentally in \cite{Lukin_2019}, that the entanglement of a time-evolved product state, measured by the von Neumann entropy of its reduced density matrix, grows in time (for large $t$) as $\log t $. This is in contrast with the behaviour of Anderson localized systems, where dynamical localization implies  that the dynamically generated entanglement does not grow in time \cite{Nach_2016}.

An effective model to explain the MBL phenomenology is the so called local integrals of motion (LIOM) model. This model assumes that a MBL Hamiltonian could be written as a sum of mutually commuting and exponentially localized terms \cite{Serbyn_2013_2, Kim_2014, Chandran_2015, Imbrie_2017, Peng_2019}. The LIOM model explains important features of a MBL, like the the dynamical generation of entanglement entropy from a pure product state as $\log t $. Moreover within the LIOM model it is possible to prove the slow spread of information, quantified by the Lieb-Robinson bounds, see equation \eqref{L-R-beta} below, \cite{Kim_2014,Nach_2021,Lu_2024}. 

The LIOM model is still under debate.  For example in the work \cite{Krajewski_2022} the authors consider perturbations  of the localized integrals of motion of the Anderson model given by interactions, finding with numerics that the component of the perturbations orthogonal to the LIOM of the Anderson model causes the system to thermalize. See also the recent work \cite{Surace_2025} that discusses the stability of the LIOM model at the perturbative level. There isn't so far an accepted proof of the existence of a complete set of exponentially localized integrals of motion for a truly interacting model. The pioneering  work of Imbrie in \cite{Imbrie_2016} relies upon an unproven assumption about a limited level attraction among energy levels of the Hamiltonian.

Different approaches are for example those of \cite{MacCormack_2021}, based on the entanglement contour, and that of the authors of \cite{Pain_2024} that obtained the $ \log \, t$-law of entanglement entropy studying 4-eigenvectors correlators as considered, for example, in the ergodic setting in the work \cite{Hahn_2024}.

Two interesting recent rigorous works are \cite{De_Roeck_2025} and \cite{Lemm_Lucas_2025}. The first one, following general ideas from \cite{Imbrie_2016}, proves MBL for subintervals of size of order $\log L $ of one-dimensional chains, this suffices to prove subdiffusion for a chain of size $L$. In the work \cite{Lemm_Lucas_2025} starting from a non-resonant conditions the authors show, in any dimension, a non-perturbatively small velocity of ballisitc information transport for all many-body states and the existence of a non-perturbatively long time scale with a logarithmically slow spread of entanglement.

In this work assuming a logarithmic lightcone, as defined in equation \eqref{L-R-beta}, but without assuming the LIOM model, we prove in theorem \ref{Renyi_from_L-R}, making use of the theory developed in \cite{Toniolo_2024_2}, a $ \log \, t$-law for the dynamical generation of $ \alpha$-Rényi entropy starting from a generic product state, with $ 0 < \alpha \le 1 $.

The scaling of $ \alpha $-Rényi entropies, $0 \le \alpha < 1 $, of the reduced density matrix of a state, of a one-dimensional system, have been shown to determine whether such a state could be efficiently represented by a matrix product state (MPS) \cite{Verstraete_Cirac_2006, Schuch_2008}. In higher dimensional systems, instead, a non volume law for the Rényi entropy does not imply an efficient MPS representation \cite{Ge_Eisert_2016}. A consequence of our theorem \ref{Renyi_from_L-R} is then that the dynamics of a local one-dimensional system with L-R bound with a logarithm lightcone let an initial pure product state to be easily simulable, under unitary evolution, for time scales of the order of the system's size $L$.

Recently Elgart and Klein \cite{Elgart_2022,Elgart_2023,Elgart_2026} have rigorously studied the $ XXZ $ model with random magnetic field, the prototypical model of choice for investigating the MBL phenomenology. One of their results, equation (2.6) of theorem 2.1 in \cite{Elgart_2026}, establishes for such system a Lieb-Robinson (L-R) bound with a logarithmic lightcone in the low energy sector of the theory. L-R bounds with logarithmic lightcones have also been established in models of fermions in a disorder potential and interacting only on a ``sparse'' set of positions in \cite{Gebert_2022,Toniolo_2024_1,Toniolo_2025_1}. L-R bounds with a log-like lightcone have been previously proven in the context of the LIOM model by \cite{Kim_2014,Nach_2021,Lu_2024}. The logarithmic lightcone of MBL dynamics has been numerically investigated in \cite{Deng_2017} and through the study of out of order time correlators in \cite{Kim_OTOC_2023}. Previous rigorous works on the XXZ model related to localization phenomena are for example \cite{Elgart_2018,Beaud_2018}. An independent derivation of the $ \log t$-law of the dynamical entanglement entropy from the assumption of a logarithmic lightcone in a local one-dimensional spin model has been given in \cite{Zeng_2023}.

It remains an important open problem whether the Lieb-Robinson bound with a logarithmic lightcone of equation \eqref{L-R-beta} of the manuscript could be extended to all energies. This would imply by the very results of our manuscript that there is a rigorous proof of the $\ln t$-law for dynamical generation of entanglement entropy for the XXZ model with random magnetic field.

Our work is organized as follows. Theorem \ref{Renyi_from_L-R} of section \ref{MBL_Rényi} is our main result. In sections \ref{LIOM_Ham} and \ref{EE_from_LIOM} we provide new proofs of established results within the LIOM model. In particular in section \ref{LIOM_Ham} we recollect the definitions of LIOM and the LIOM Hamiltonian. From the LIOM Hamiltonian we prove in lemma \ref{theo_L-R_LIOM}, of section \ref{sub_L-R}, a L-R bound with a logarithmic lightcone. In section \ref{EE_from_LIOM} assuming the LIOM model we prove in lemma \ref{lem_EE_from_LIOM} the $ \log \, t$-law of entanglement entropy. We conclude with section \ref{Discussion} discussing and comparing our results on dynamical entanglement entropy from sections \ref{MBL_Rényi} and \ref{EE_from_LIOM}.



\section{Dynamical generation of $\alpha$-Rényi Entropies in systems with logarithmic lightcones} \label{MBL_Rényi}

We consider a one dimensional lattice $ \Lambda = [-L,L] \cap \mathds{Z} $ of qubits (the on site Hilbert space is $ \mathds{C}^2 $). The theory can be generalized to qudits as well. The system's Hilbert space is $ \mathcal{H} := \left(\mathds{C}^2\right)^{\otimes 2L+1} $.

The $ \alpha $-Rényi entropy, with $ \alpha \neq 1 $, given $ \rho: \mathcal{H} \rightarrow \mathcal{H} $  a state (that means $ \rho \ge 0 $ and $ \Tr \rho = 1 $), is defined as:
\begin{equation}
 S_\alpha (\rho) = \frac{1}{1-\alpha} \log \Tr \rho^\alpha
\end{equation}
It holds $ S_\alpha (\cdot) \le S_\beta (\cdot) $ with $ \alpha \ge \beta $. The von Neumann entropy, $ S(\rho)=-\Tr \rho \log \rho $, is obtained as the $ \alpha \rightarrow 1 $ limit of the $ \alpha $-Rényi entropies (see for example equations 18-19 of \cite{Toniolo_2024_2}). We consider $ \alpha $-Rényi entropies with $ \alpha \le 1 $. This means that our attention is focused on quantum entropies that upper bound the von Neumann entropy.

In this work we use the formalism developed by us in \cite{Toniolo_2024_2} to upper bound 
\begin{equation} \label{delta_Renyi}
 \Delta \, S_{\alpha}(t) := \Eom \left( S_{\alpha}\left( \Tr_{[1,L]}  e^{-it \, H_\om} \rho e^{it\, H_\om} \right) - S_{\alpha}\left( \Tr_{[1,L]} \rho  \right) \right)
\end{equation}
with $ \rho $ a product state $ \rho=\otimes_{j=-L}^L \rho_j $, $ \rho_j : \mathds{C}^2 \rightarrow \mathds{C}^2 $, and a nearest neighbour Hamiltonian $ H_\om := \sum_{r=-L}^{L-1} H_{r,r+1,\om} $, dependent on a parameter $ \om $, that usually refers to disorder. We assume that the dynamics $ U_\om(t):=e^{-itH_\om} $ gives rise to a Lieb-Robinson bound with a logarithmic lightcone. This means that there exist $ K > 0 $, $ \beta > 0 $ and $ \mu > 0 $ (the scaling, or localization, length) all system's size independent, such that for every pair of operators $ A $ and $ B $ defined on $ \mathcal{H} $, $ A,B : \mathcal{H} \rightarrow \mathcal{H} $ with support that for simplicity is assumed connected, it holds:
\begin{align} \label{L-R-beta}
& \Eom \| [  U_\om(t) A U_\om^*(t) , B ] \| \le  K \, \|A\| \, \|B\| \, t^\beta \, e^{-\frac{\dist(\supp(A),\supp(B))}{\mu} }
\end{align}
The operator norm $\|A\|$ is the largest singular value of $A$. Denoting $ X $ an interval containing the support of $ A $ and $ X^c $ the complement of $ X $ in $ [-L,L] $ it holds \cite{Hastings_Locality_2010} that
\begin{align} \label{Haar}
 &\frac{1}{2^{|X^c|}} \left( \Tr_{X^c} U^* A U \right) \otimes \mathds{1}_{X^c}  = \int_{W \textrm{supported on}\, X^c}  \, W^* (U^* A U) W \, d_{Haar}W
\end{align} 
Using \eqref{Haar} it is easy to see that \eqref{L-R-beta} implies
\begin{align} \label{L-R-beta-res}
& \Eom \| \frac{1}{2^{|X^c|}} \left( \Tr_{X^c} U_\om^*(t) A U_\om(t) \right) \otimes \mathds{1}_{X^c} -  U_\om^*(t) A U_\om(t) \| \le  K   \, t^\beta  \, \|A\| \,  e^{-\frac{\dist(\supp(A),X^c)}{\mu} }
\end{align}
In the following we denote $ l:= \dist(\supp(A),X^c) $.

Let us discuss the meaning of a ``logarithmic lightcone'', in the simplest cases when $ \beta=1 $, then  $ K\,t\,   e^{ -\frac{l}{\mu} } =   e^{\log(Kt) -\frac{l}{\mu} } $ the support of $ A $ is $ x=0 $. This means that up to a time $ t $ exponentially large in $ l $ the restriction of $ \Eom U_\om^*(t) A U_\om(t) $ within the region $ [-l,l] $ provides an approximation exponentially good to the 
full dynamics $ \Eom U_\om^*(t) A U_\om(t) $. For example, defining $ K t_{max} = e^{\frac{l}{2\mu}} $, at $ t=t_{max} $ the upper bound in \eqref{L-R-beta} is of the order of $ e^{-\frac{l}{2\mu}} \ll 1 $, when $ l \gg \mu $. For comparison, the L-R bound of an ergodic system $ e^{v_{LR} t -\frac{l}{\eta} } $ is such that the truncation of the dynamics provides a good approximation only within times proportional to $ l $. This implies that the dynamics of a many-body localized system, assuming that the L-R bound \eqref{L-R-beta} holds, is exponentially slower than the ergodic one, but not completely ``frozen'' as it would result from a vanishing L-R velocity, $ v_{LR}=0 $, given in equation \eqref{Anderson_type_in}, as proven in disordered non-interacting systems \cite{Sims_Stolz_2012}, see also \cite{Sims_2016}. $ v_{LR}=0 $  in the single particle context corresponds to the dynamical localization in the Anderson model \cite{Stolz_2011}.

The L-R bound \eqref{L-R-beta} is meaningful within times $ t $ such that the RHS of \eqref{L-R-beta} is smaller than the trivial bound that equals $ 2 \| A \| $. The same is true for \eqref{L-R-beta-res}, where the trivial bound follows, for example, from \eqref{Haar}.

The case $ \beta=0 $ in \eqref{L-R-beta} is relevant with respect to Anderson localization. 
The authors of \cite{Sims_Stolz_2012}, see their corollary 4.1, showed that for the dynamics of the $ XY $ model with random magnetic field, that is equivalent to a non-interacting many-particle Anderson model, the following holds: there are parameters $ c, \mu > 0 $ such that
\begin{align} \label{Anderson_type_in}
&  \Eom \sup_t \| \frac{1}{2^{|X^c|}} \left( \Tr_{X^c} U_\om^*(t) A U_\om(t) \right) \otimes \mathds{1}_{X^c} -  U_\om^*(t) A U_\om(t) \|  \le  c  \, \|A\| \,  e^{-\frac{l}{\mu} }
\end{align}
Note the supremum with respect to $ t $ on the LHS of \eqref{Anderson_type_in}, and also that both the left-hand sides of \eqref{Anderson_type_in} and \eqref{L-R-beta} involve an average over all realisations of the Hamiltonian $H_\om$. We will refer to \eqref{Anderson_type_in} as Anderson-type localization. For  the $ XY $ model with random magnetic field in the work \cite{Nach_2016} it has been established an upper bound on $ \Delta \, S(t) $, with the supremum over time taken before $ \Eom $, like in \eqref{Anderson_type_in}, that is time-independent. We also like to mention the works \cite{Pastur_2014,Elgart_2016,Pastur_2018} that studied the entanglement properties of the eigenstates of the Anderson model in various dimensions, in particular \cite{Pastur_2018} established the absence of self-averaging for the entanglement entropy of the reduced Fermi projection of the Anderson model in a one-dimensional lattice.

It can be argued that a system displaying Anderson-type localization, as defined in \eqref{Anderson_type_in} would also satisfy \eqref{L-R-beta}, at least for $ t $ large enough. This implies that a genuine MBL phase cannot simply be defined by \eqref{L-R-beta}. It is clear that, in this setting, to distinguish among Anderson-type localization and MBL a lower bound to the Lieb-Robinson commutator (namely the LHS of \eqref{L-R-beta}), or to its time-average till time $t$, increasing  in $ t $ (and compatible with the upper bound) is needed. We think that some type of delocalization of information must be proven to actually distinguish a genuine many-body localized phase from Anderson localization.

 The authors of \cite{Linden_2009} found out that every many-body system with an energy spectrum without degeneracies and without degenerate energy gaps equilibrates. Later on the authors of  \cite{Short_2012} partially lift the condition on the absence of degenerate eigenvalues and degenerate gaps.
Models displaying Anderson-type localization do not have such a spectrum, at least in the thermodynamic limit \cite{Minami_1996,Stolz_2011,Sims_Stolz_2012}. It has been recently claimed in \cite{Huang_2021} that the spectral conditions assumed in \cite{Linden_2009} generically holds for local interacting Hamiltonians, namely that the set of systems that violate these spectral features have in the space of parameters, that define the model, zero measure.

A remark about units. We are working in natural units, $ \textrm{energy} \times \textrm{time}=1 $, therefore $ K $ in \eqref{L-R-beta} is such that $ K t^\beta $ is dimensionless. All the distances are measured in units of the lattice spacing, that is taken equal to $ 1 $, in fact $ \Lambda := [-L,L] \cap \mathds{Z} $, therefore $ \mu $ is dimensionless.

The input of the theory developed by us in \cite{Toniolo_2024_2} is the Lieb-Robinson bound of the dynamics, even in averaged form, as in \eqref{L-R-beta}. The output is an upper bound on the dynamical Rényi entropy, as defined in \eqref{delta_Renyi}, starting from a generic product state $ \rho$. The $\alpha$-Rényi entropy with $0 \le\alpha  <1$ is an important quantity because it characterizes the states that can be efficiently approximated by matrix product states \cite{Verstraete_Cirac_2006, Schuch_2008}.

It is important to mention that Elgart and Klein \cite{Elgart_2023,Elgart_2022} rigorously obtained for the  $XXZ$ model, with strong enough on site disorder and in the low energy sector of the theory, a Lieb-Robinson bound as in \eqref{L-R-beta} but with an extra factor dependent on the system's size $ L^{\xi_E} $. For the precise statement of their bound we refer to equation 1.1 on page 3 of their work \cite{Elgart_2023}. In this case the meaning of the logarithmic lightcone discussed before requires the additional condition $ l \gg \ln L $.




\begin{Theorem} \label{Renyi_from_L-R}
Given a local Hamiltonian, for simplicity nearest neighbour $ H_\om := \sum_{r=-L}^{L-1} H_{r,r+1,\om} $, where the dependence of $ H_{r,r+1,\om} $ on $ \om $ is through a scalar function $ f $, $H_{r,r+1,\om} = A_{r,r+1} + f(\om_r,\om_{r+1}) B_{r,r+1}$, see the example \eqref{example}, and with $ J := \max_{r} \|A_{r,r+1} \| + \max_{r,\om}|f(\om_r,\om_{r+1})| \max_r \|B_{r,r+1} \| $, that is assumed to give rise to a Lieb-Robinson bound with a logarithmic lightcone, as in \eqref{L-R-beta}, then the variation of $ \alpha$-Rényi entropy \eqref{delta_Renyi}, starting from any initial product state $ \rho $, is upper bounded for long times, with $ \frac{\ln2}{\ln2 + 1/\mu} < \alpha \le 1 $, and $ \eta > 1 $, as follows:
\begin{align}\label{delta_Renyi_time}
 \lim_{t'\rightarrow \infty} \lim_{L\rightarrow \infty} \frac{ \Eom \, \Delta S_\alpha (t')}{\ln t'} \le \eta \frac{\beta +2}{\frac{1}{\mu}-\frac{1-\alpha}{\alpha}\ln2  }
\end{align}
$ t'$ is a (dimensionless) rescaled time, $ t':= \left( \frac{8KJ^2}{(\beta +1)(\beta+2)} \right)^{\frac{1}{\beta +2}}  t $. The deviation of $ \Delta S_\alpha (t') $ from $ \ln t' $ decreases as a power law in $ t' $, see equation \eqref{30}.
\end{Theorem}

Before the proof let us provide some remarks.

We stress that the inputs needed to obtain \eqref{delta_Renyi_time} are the locality of the Hamiltonian  and its L-R bound in the form given by \eqref{L-R-beta}. This result does not depend on the explicit form of the Hamiltonian, moreover we are not assuming the existence of a complete set of localized integrals of motion (LIOM). In Lemma \ref{lem_EE_from_LIOM} of section \ref{EE_from_LIOM}, we obtain for $ \Delta S_\alpha (t') $ the same $ \ln t' $ scaling assuming instead the existence of LIOM.

As mentioned in the introduction theorem \ref{Renyi_from_L-R} implies that the initial pure product state is easily simulable, under unitary evolution, for time scales of the order of the system's size $L$.
More precisely: if there is a $0 \le \alpha < 1 $ such that the reduced density matrix of a pure state $ |v\rangle $ satisfies $ S_\alpha(\Tr_{[1,L]} |v\rangle \langle v| ) \le O(\log L) $ then the bond dimension (D) of the MPS representation of $ |v\rangle $ is of order $ \textrm{poly}(L) $, meaning that such representation is efficient anf the state is easily (classically) simulable.

Our result \eqref{delta_Renyi_time} could be generalized to Hamiltonians with exponential decay of interactions following appendix G of \cite{Toniolo_2024_2}. The meaning of the dependence of the bound \eqref{delta_Renyi_time} on $ \mu $ is straightforward: the smaller $ \mu $ the slower the dynamics, in fact the time scale for the spread of the support of an operator in a system with a logarithmic lightcone is proportional to $ e^{\frac{l}{\mu}} $. This corresponds to a slower growth of entropy, in fact the RHS of \eqref{delta_Renyi_time} decreases with $ \mu $ decreasing.

\begin{proof}
This proof starts with the following equation from the introduction of \cite{Toniolo_2024_2}, for the details see section IV of \cite{Toniolo_2024_2}. For now we consider $ H $ independent from $ \om $. We go back to it in equation \eqref{def_Delta} where we explain how the average $ \Eom $ is taken into account. 
Defining $ H_{\Lambda_k} := \sum_{j=-k}^{k-1} H_{j,j+1} $, and 
\begin{align} \label{def_Delta_gen}
  \Delta_k(t) & := \big \| e^{is(H_{\Lambda_{k+1}}-H_{0,1})}H_{[0,1]}e^{-is(H_{\Lambda_{k+1}}-H_{[0,1]})} - e^{is(H_{\Lambda_{k}}-H_{[0,1]})}H_{[0,1]}e^{-is(H_{\Lambda_{k  x}}-H_{[0,1]})}  \big \| \\
  & \le  \int_0^t ds  \big \| \big[(H_{k,k+1}+H_{-k-1,-k}) \,,\, e^{is(H_{\Lambda_k}-H_{[0,1]})}H_{[0,1]}e^{-is(H_{\Lambda_k}-H_{[0,1]})} \big] \big \| 
\end{align}
according to \cite{Toniolo_2024_2}, where $ l $ is a distance to be determined by minimization, it holds
\begin{align} \label{bound_sum_1}
  \Delta S_\alpha & (t) \le  \frac{1}{1-\alpha} \sum_{k=l+1}^{L} \log \Big [ 1- \alpha \int_0^t ds \, \Delta_{k-1}(s) + (2^{k+1}-1)^{1-\alpha} \left(\int_0^t ds \, \Delta_{k-1}(s)\right)^\alpha \Big   ] + l +1 
\end{align}
Let us sketch how \eqref{bound_sum_1} is obtained.

The main idea for the evaluation of the dynamically generated entanglement is, according to \cite{Eisert_Osborne_2006, Osborne_2006}, to replace the unitary evolution $ e^{-itH} $ with $ V(t):= e^{it(H_L+H_R)} e^{-itH} $, having introduced the ``left`` and ''right`` Hamiltonians $ H_L:= \sum_{k=-L}^{-1} H_{k,k+1} $ and $ H_R:= \sum_{k=1}^{L-1} H_{k,k+1} $. Neither $ H_L $ nor $ H_R $ contain the term $ H_{0,1} $ connecting the two halves $ [-L,0] $ and $ [1,L] $ of the system's bipartition. This leaves the dynamical entropy invariant, in fact:
\begin{align}
& S_\alpha \big( \Tr_{[1,L]} V(t) \rho V^*(t) \big)   = S_\alpha  \big( \Tr_{[1,L]} e^{it (H_L+H_R)} e^{-itH} \rho e^{itH} e^{-it (H_L+H_R)}  \big)   \\
& = S_\alpha \big( \Tr_{[1,L]} e^{it H_L} e^{itH_R} e^{-itH} \rho e^{itH} e^{-itH_R} e^{-it H_L}  \big)  =S_\alpha \big[ e^{it H_L} \big(  \Tr_{[1,L]}  e^{itH_R} e^{-itH} \rho e^{itH} e^{-itH_R}  \big) e^{-it H_L} \big] \label{V11}  \\
&=S_\alpha \big(  \Tr_{[1,L]}  e^{itH_R} e^{-itH} \rho  e^{itH} e^{-itH_R} \big) =S_\alpha \big(  \Tr_{[1,L]} e^{-itH} \rho e^{itH} \big) \label{V13} 
\end{align}
The physical intuition about $ V(t) $ is that despite being supported on the whole lattice, at short times only the terms in the vicinity of the lattice site $ x=0 $ are relevant, therefore $ V(t) $ can be approximated, up to a small error, by a factor with the same structure but Hamiltonians $ H_L $, $H_R$ and $ H $ restricted to the interval $ [-(L-1),L-1] $. This suggests to upper bound $ S_{\alpha}\left( \Tr_{[1,L]}  e^{itH} \rho e^{-itH} \right) - S_{\alpha}\left( \Tr_{[1,L]} \rho  \right) $ with a telescopic sum, where each pair of terms is in turn upper bounded using the general, and tight, inequality proven by Audenaert, see equation A3 of \cite{Audenaert_2007}. We have employed a slightly weaker form of Audenaert's bound given below in \eqref{Audenaert_Datta}. With  $ T := \frac{1}{2}\|\rho-\sigma\|_1 $ the trace distance among two states $ \rho $ and $ \sigma $, $ d $ the dimension of the Hilbert space where the states are acting upon, it holds:
\begin{equation} \label{Audenaert_Datta}
  |S_\alpha(\rho)-S_\alpha(\sigma)| \le 
\frac{1}{1-\alpha} \log ( 1-\alpha T + (d-1)^{1-\alpha} T^\alpha) 
\end{equation}
 The bound \eqref{Audenaert_Datta} is increasing in $ T $, therefore it still holds true replacing $T$ with $R$, being  $ T \le R \le 1 $. This will be useful in the following because we can only provide upper bounds to the trace distance $T$. 
 
In our theory the state $ \rho $ appearing in \eqref{delta_Renyi} is crucially assumed being a product state $ \rho=\otimes_{j=-L}^L \rho_j $, with $ \rho_j : \mathds{C}^2 \rightarrow \mathds{C}^2 $. Nevertheless the theory could also be extended low entangled states, see Lemma 3 of \cite{Toniolo_2024_2}. 
The telescopic sum stops with the restriction of the Hamiltonian to the interval $ [-l,l] $. The whole sum in \eqref{bound_sum_1} is then supposed to be minimized with respect to $ l $. This minimization has been performed in \cite{Toniolo_2024_2} in the case of a linear lightcone.
Equation \eqref{bound_sum_1} then arises from the application of \eqref{Audenaert_Datta} to each pair of terms of the telescopic sum where an upper bound to the trace distance is obtained making use of L-R bounds.
In this proof the minimization of the telescopic sum is performed explicitly, having assumed the L-R bounds \eqref{L-R-beta}, in equations \eqref{27}-\eqref{l_min} below.

Before that we explicitly introduce the ''disorder`` parameter $ \om $. The definition of $ \Delta S_\alpha(t) $ in \eqref{delta_Renyi} includes the average over $ \om $. This average is brought inside the $ \log $ in equation \eqref{bound_sum_1} first, and then under the power to the $ \alpha $ in the last term of the argument of the $ \log $ using the the Jensen inequality, in fact both the logarithm and the power $ 0< \alpha \le1$ are concave functions. The Jensen inequality for concave functions reads \cite{Rudin_Principles}: given $ p_j \in (a,b) $, with $ p_j > 0 $ and $ 1 = \sum_j p_j $, if $ f $ is concave in $ (a,b) $, given $ x_j \in (a,b) $ it holds:
\begin{align} \label{Jensen}
 \sum_j p_j f(x_j) \le f \left( \sum_j p_j x_j \right)
\end{align}
In our work the ''weights`` $ p_j $ are hidden in the normalized average $ \Eom $. The average over $ \om $ is such that the RHS of \eqref{def_Delta_gen}, also with $ k \rightarrow k-1 $, is replaced with:
\begin{align} \label{def_Delta}
 & \Delta_{k-1}(t) \le  \int_0^t ds \Eom \big \| \big[(H_{k-1,k,\om}+H_{-k,-k+1,\om}) \,,\, e^{is(H_{\Lambda_{k-1},\om}-H_{[0,1],\om})}H_{[0,1],\om}e^{-is(H_{\Lambda_{k-1},\om}-H_{[0,1],\om})} \big] \big \| 
\end{align}
We are now in the position to use the L-R bound as in equation \eqref{L-R-beta} once we have taken into account that the operators in the commutator on the RHS of \eqref{def_Delta} are disorder-dependent. We recollect the assumption about the scalar dependence of $ H_{k,k+1,\om} $ on $ \om $:
\begin{align} \label{scalar_dis}
 H_{k,k+1,\om} = A_{k,k+1} + f(\om_k,\om_{k+1}) B_{k,k+1}
\end{align}
where $ A_{k,k+1} $ and $ B_{k,k+1} $ are supported on the sites $ k $ and $ k+1 $. 
In the notable example of the $ XY $ model with random magnetic field \cite{Sims_Stolz_2012}, it is:
\begin{align} \label{example}
\begin{cases}
 & A_{k,k+1} = \sigma_k^x \sigma_{k+1}^x + \sigma_k^y \sigma_{k+1}^y  \\
 &  f(\om_k,\om_{k+1}) B_{k,k+1} = \omega_k \sigma_k^z + \omega_{k+1} \sigma_{k+1}^z
\end{cases}
\end{align}
Then, defining $ U_{\Lambda_{k},\om}(s) := e^{-is(H_{\Lambda_{k},\om}-H_{[0,1],\om})} $, we have
\begin{align}
& \Eom \| \big[H_{k-1,k,\om} \,,\, U_{\Lambda_{k},\om}^*(s)H_{[0,1],\om}U_{\Lambda_{k},\om}(s) \big] \| \\
& \le \Eom \| \big[A_{k-1,k} \,,\, U_{\Lambda_{k},\om}^*(s)A_{[0,1]}U_{\Lambda_{k},\om}(s) \big] +  f(\om_{k-1},\om_{k})  \big[B_{k-1,k} \,,\, U_{\Lambda_{k},\om}^*(s)A_{[0,1]}U_{\Lambda_{k},\om}(s) \big] \| \nonumber \\
& + \Eom \| f(\om_{0},\om_{1})  \big[A_{k-1,k} \,,\, U_{\Lambda_{k},\om}^*(s)B_{[0,1]}U_{\Lambda_{k},\om}(s) \big] + f(\om_{k-1},\om_{k}) f(\om_{0},\om_{1})  \big[B_{k-1,k} \,,\, U_{\Lambda_{k},\om}^*(s)B_{[0,1]}U_{\Lambda_{k},\om}(s) \big] \|  \\ 
& \le \Eom \| \big[A_{k-1,k} \,,\, U_{\Lambda_{k},\om}^*(s)A_{[0,1]}U_{\Lambda_{k},\om}(s) \big] \| \nonumber \\
& +  \max_{k\in[-L,L],\om}|f(\om_{k-1},\om_{k})|  \Eom \|  \big[B_{k-1,k} \,,\, U_{\Lambda_{k},\om}^*(s)A_{[0,1]}U_{\Lambda_{k},\om}(s) \big] \| \nonumber \\
& + \max_\om |f(\om_{0},\om_{1})| \Eom \|  \big[A_{k-1,k} \,,\, U_{\Lambda_{k},\om}^*(s)B_{[0,1]}U_{\Lambda_{k},\om}(s) \big] \| \nonumber \\
& + \max_{k\in[-L,L],\om}|f(\om_{k-1},\om_{k}) f(\om_{0},\om_{1}) | \Eom \| \big[B_{k-1,k} \,,\, U_{\Lambda_{k},\om}^*(s)B_{[0,1]}U_{\Lambda_{k},\om}(s) \big] \|  \\ 
& \le 4K \, t^\beta \, J^2  \,  e^{ -\frac{k-2}{\mu} } \label{16}
 \end{align}
In \eqref{16} we have used the fact that $ \dist(\supp(H_{[k-1,k],\om}),\supp H_{[0,1],\om}) = k-2 $ and the definition of $ J $ in \ref{Renyi_from_L-R}. The contribution from the term $ H_{[-k,-k+1],\om} $ in \eqref{def_Delta} is the same as in \eqref{16}. Performing the time integral in the RHS of \eqref{def_Delta}, we obtain:
\begin{align} \label{up_Delta}
 \Delta_{k-1}(t) \le \frac{8KJ^2}{\beta +1}t^{\beta +1}  e^{-\frac{k-2}{\mu}} 
\end{align}
Inserting \eqref{up_Delta} into \eqref{bound_sum_1} and defining the (dimensionless) rescaled time $ t'$
\begin{equation} \label{t'_def}
 t':= \left( \frac{8KJ^2}{(\beta +1)(\beta+2)} \right)^{\frac{1}{\beta +2}}  t
\end{equation}
we get
\begin{align} \label{bound_sum_2}
 \Eom \, \Delta S_\alpha (t') & \le \frac{1}{1-\alpha} \sum_{k=l+1}^{L} \log \Big [ 1- \alpha t'^{\beta +2}  e^{-\frac{k-2}{\mu}}  + (2^{k+1}-1)^{1-\alpha} (t'^{\beta +2}  e^{-\frac{k-2}{\mu}})^\alpha \Big ] + l +1 
\end{align}
Rewriting $ 2^{k(1-\alpha)}  e^{-\frac{k \alpha}{\mu}} = e^{ \left((1-\alpha)\ln 2 -\frac{\alpha}{\mu} \right)k}$, it follows that the convergence of the series in \eqref{bound_sum_2}, in the limit $ L \rightarrow \infty $, is ensured by: 
\begin{equation} \label{summ_cond}
 (1-\alpha)\ln2 - \frac{\alpha}{\mu} < 0 \hspace{5mm} \Rightarrow \hspace{5mm}  \frac{\ln 2}{\ln 2 + \frac{1}{\mu}} <  \alpha \le 1 
\end{equation}
At this point we adopt a different approach with respect to the one we employed in \cite{Toniolo_2024_2}, that was based on the approximate solution of the system given in equation (39) of \cite{Toniolo_2024_2}. We will explicitly find an upper bound to $ \Eom \, \Delta S_\alpha (t') $ summing up the series in \eqref{bound_sum_2}, and then discussing the minimization with respect to $ l $ as a function of time.

Defining $ \lambda := \frac{\alpha}{\mu} - (1-\alpha)\ln2 $ and then changing the variable of the sum into $ k:= n+l $, it is:
\begin{align} 
 \Eom \, \Delta S_\alpha (t') & \le \frac{1}{1-\alpha} \sum_{k=l+1}^{\infty}  \Big [ - \alpha t'^{\beta +2} e^{\frac{2}{\mu}}  e^{-\frac{k}{\mu}}  + 2^{1-\alpha} e^{\frac{2\alpha}{\mu}} t'^{\alpha(\beta +2)}  e^{-\lambda k} \Big ] + l +1 \label{27} \\
 & = \frac{1}{1-\alpha} \sum_{n=1}^{\infty}  \Big [ - \alpha t'^{\beta +2} e^{\frac{2}{\mu}}  e^{-\frac{n+l}{\mu}}  + 2^{1-\alpha} e^{\frac{2\alpha}{\mu}} t'^{\alpha(\beta +2)}  e^{-\lambda (n+l)} \Big ] + l +1 \label{28} \\
 & = \frac{1}{1-\alpha} \Big [ - \alpha t'^{\beta +2} e^{-\frac{l-2}{\mu}}  \frac{1}{e^{\frac{1}{\mu}}-1}  + 2^{1-\alpha} e^{\frac{2\alpha}{\mu}} t'^{\alpha(\beta +2)}  e^{-\lambda l} \frac{1}{e^{\lambda}-1} \Big ] + l +1  \label{30} 
\end{align}
A simple way of minimizing, at least for large $ t' $, the upper bound in \eqref{30} with respect to $ l $ is to ask which is the ''smallest`` $ l $, as a function of $ t' $, that ensures the term inside the parenthesis in \eqref{30} to go to zero for $ t' \rightarrow \infty $. Then the term $ l $ outside of the parenthesis in \eqref{30} will determine the long time behaviour of $ \Eom \, \Delta S_\alpha (t') $.
Asking that 
\begin{align} \label{condition_asym}
\lim_{t' \rightarrow \infty} t'^{\alpha(\beta +2)}  e^{-\lambda l} \rightarrow 0 
\end{align}
also ensures that  $ \lim_{t' \rightarrow \infty} t'^{\beta +2} e^{-\frac{l}{\mu}} \rightarrow 0 $, in fact it is $ \frac{\lambda}{\alpha} < \frac{1}{\mu} $.
Then, using $ \lambda := \frac{\alpha}{\mu} - (1-\alpha)\ln2 $, we choose
\begin{align} \label{l_condition}
 l >  \frac{\beta +2 }{\frac{1}{\mu} - \left(\frac{1}{\alpha} - 1 \right)\ln2  } \ln t'
\end{align}
leading, with $ \eta > 1 $, to the definition
\begin{align} \label{l_min}
 \lm := \eta \frac{\beta +2 }{\frac{1}{\mu} - \left(\frac{1}{\alpha} - 1 \right)\ln2  } \ln t'
\end{align}
Other possible choices for $ l(t') $ could have been a power law in $ t' $, that satisfies \eqref{condition_asym} but gives a worse bound than the one associated with \eqref{l_min}, in fact $ \forall a > 0 $, $ \lim_{x \rightarrow \infty} \frac{x^a}{\ln x} = \infty $. On the other hand a choice of $ l(t') $ with a slower asymptotic increase at $ \infty $ than $ \ln t' $, like $ (\ln t')^b $, with $ 0 < b < 1 $, because of the factor $ t'^{\alpha(\beta+2)} $ would give again a worse upper bound than $ \ln t' $.

We can conclude that with $ \frac{\ln 2}{\ln 2 + \frac{1}{\mu}} < \alpha < 1 $:
\begin{align} 
 \lim_{t'\rightarrow \infty} \lim_{L\rightarrow \infty} \frac{\Eom \,\Delta S_\alpha (t')}{\ln t'} \le \eta \frac{\beta +2}{\frac{1}{\mu}-\frac{1-\alpha}{\alpha}\ln2  }
\end{align}

 The power law correction, that goes to zero with $t' \rightarrow \infty $, given by the  terms inside the parenthesis in \eqref{30}, with $ l = \lm $ as in equation \eqref{l_min}, is given by:
\begin{align} \label{correction}
 \frac{1}{1-\alpha} \Big [ - \alpha \frac{e^{\frac{2}{\mu}}}{e^{\frac{1}{\mu}}-1} t'^{-(\beta +2)\left( \frac{\alpha \eta}{\mu \lambda}-1 \right)} + 2^{1-\alpha} \frac{e^{\frac{2\alpha}{\mu}}}{e^{\lambda}-1} t'^{-\alpha(\beta +2)(\eta-1)} \Big ]
\end{align}
We remark that $\frac{\alpha}{\mu \lambda} > 1 $. We see that despite having chosen $ \eta > 1 $, fixing $ \eta = 1 $ would simply correspond in \eqref{correction} to a constant contribution (from the second term) that with $t' \rightarrow \infty $ is still sub-leading with respect to $\ln t'$.

Let us discuss the limit $ \alpha \rightarrow 1 $ in \eqref{30}.
The factor in between square brackets in \eqref{30} is a differentiable function of $ \alpha $ and is vanishing with $ \alpha = 1$. This means that the limit $ \alpha \rightarrow 1 $ of the first term in \eqref{30} equals minus the first derivative in $ \alpha = 1$ of such function. When $ l $ is replaced with $ \lm $ as in \eqref{l_min} an explicit calculation reveals that the limit $ \alpha \rightarrow 1 $ of the first term of \eqref{30} is decreasing in $ t' $ as a power law, this means that the long $ t' $ behaviour of  $ \Eom \, \Delta S_\alpha (t') $ is again given by $ \ln(t') $. In particular we have that the limits $ \alpha \rightarrow 1 $ and $ t' \rightarrow \infty $ exchange themselves.

\begin{align} 
 \lim_{t'\rightarrow \infty} \lim_{\alpha \rightarrow 1} \lim_{L\rightarrow \infty} \frac{\Eom \,\Delta S_\alpha (t')}{\ln t'} = \lim_{\alpha \rightarrow 1} \lim_{t'\rightarrow \infty}  \lim_{L\rightarrow \infty} \frac{\Eom \,\Delta S_\alpha (t')}{\ln t'} \le \lim_{\alpha \rightarrow 1} \eta \frac{\beta +2}{\frac{1}{\mu}-\frac{1-\alpha}{\alpha}\ln2  } = \eta \mu(\beta +2)
\end{align}

\end{proof}


\section{Local Integral of Motion Hamiltonian} \label{LIOM_Ham}

Given any local Hamiltonian of the type $ H_\om=\sum_j h_{j,\om} $ with $ \{ h_{j,\om} \}$ supported on a finite and system size independent region around $ j $, following \cite{Kim_2014} we can rewrite this Hamiltonian as a sum of commuting operators $\{ H_{j,\om} \} $: 
\begin{align} \label{cons_comm}
 H_{j,\om} := \lim_{T \rightarrow \infty} \frac{1}{T} \int_0^T dt e^{itH_\om} h_{j,\om} e^{-itH_\om}
\end{align}
$ H_{j,\om} $ are, in principle, supported on the all system. It is easy to see that
\begin{align} \label{local_LIOM}
 H_\om = \lim_{T \rightarrow \infty} \frac{1}{T} \int_0^T dt e^{itH_\om} H_\om e^{-itH_\om} = \lim_{T \rightarrow \infty} \frac{1}{T} \int_0^T dt e^{itH_\om} \sum_j h_{j,\om} e^{-itH_\om} = \sum_j H_{j,\om}
\end{align}
and that $ [H_{j,\om},H_{k,\om}]=0$. This follows, for example, from evaluating the average of a given $ H_j $ with respect to two eigenstates $ \psi_a $ and $ \psi_b $ of the full Hamiltonian $ H_\om $ corresponding to two distinct energies $ E_a $ and $ E_b $:
\begin{align}
 \langle \psi_a, H_{j,\om} \psi_b \rangle = \lim_{T \rightarrow \infty} \frac{1}{T} \int_0^T dt e^{it(E_a-E_b)} \langle  \psi_a , h_{j,\om}  \psi_b \rangle = \delta(E_a,E_b) \langle  \psi_a , h_{j,\om}  \psi_b \rangle
\end{align}
Assuming that the eigenvalues of $ H_\om $ are all distinct (recently it has been claimed in \cite{Huang_2021} that, for a generic local Hamiltonian, it is actually true more namely that also the energy gaps are non degenerate) this shows that the basis $ \{\psi_a\} $ diagonalizes each $ H_{j,\om} $. The $ \{ H_{j,\om} \} $ then share a complete basis of eigenvectors, therefore they commute with each other, and in turn with $ H_\om $.

The crucial assumption made within the theory of LIOM (local integrals of motion) is claiming that each $ H_{j,\om} $ is exponentially localized in an averaged sense \cite{Chandran_2015}.
In the literature other methods have been used to introduce LIOM, see for example \cite{Peng_2019}, and the discussion in \cite{Lu_2024} and references therein.

The LIOM Hamiltonian is:
\begin{align} \label{LIOM}
 H_\om = \sum_{r=-L}^L H_{r,\om}
\end{align}
The $\{ H_{r,\om} \} $ are supported on the whole lattice, but with exponential tails. Let us formalize this as follows. $ H_{r,\om}: \left(\mathds{C}^2\right)^{\otimes 2L+1} \rightarrow \left(\mathds{C}^2\right)^{\otimes 2L+1} $. We will refer to $ H_{r,\om} $ as ``centered'' in $ r $. Along the lines of the previous section \ref{MBL_Rényi}, where we were considering a nearest neighbour Hamiltonian, we define $ J $ as the maximal local energy, $ J := \max_\om \max_r \| H_{r,\om} \| $. The set $\{ H_{r,\om} \} $ has two properties:
\begin{itemize}
 \item $ [H_{r,\om},H_{s,\om}]=0 $
 \item Exponential ``tails''. With $ r $ inside the region $ X $, that for simplicity we assume connected, it holds:
 \begin{equation} \label{decay}
\Eom \, \| \frac{1}{2^{|X^c|}} \left( \Tr_{X^c} H_{r,\om} \right) \otimes \mathds{1}_{X^c} - H_{r,\om} \| \le  J \,  e^{ -\frac{\textrm{dist}(r,X^c)}{\xi} }
\end{equation}
This means that restricting the support of $ H_{r,\om} $ to the region $ X $, with $ r \in X $, the difference with $ H_{r,\om} $, in norm, decreases exponentially with the distance of $ r $ from the complement of $ X $. $ X^c $ denotes the complement of $ X $ in $ [-L,L] $. In few words: the larger the distance of $ r \in X $ from the edge of $ X $, the better the approximation of $ H_r $ with the restriction of $ H_{r,\om} $ to the region $ X $. $ \Eom $ denotes averaging with respect to all possible realisations of the Hamiltonian associated to different disorder configurations.
\end{itemize}

The authors of \cite{Sims_Stolz_2012} showed that for the dynamics of the disordered $ XY $ model, that is equivalent to a non-interacting many-particle Anderson model, the following holds. Given a local observable $O$ with support at the origin of the lattice, let $ (e^{iH_\om t} O e^{-iH_\om t})_l $ be the restriction of $ e^{iH_\om t} O e^{-iH_\om t} $ onto the ball of radius $ l>0 $ around the origin.
There are parameters $ c, \mu > 0 $ such that
\begin{align}\label{Anderson_type}
 \Eom  \sup_t \| (e^{iH_\om t} O e^{-iH_\om t})_l - e^{iH_\om t} O e^{-iH_\om t} \| \le c \, e^{-\frac{l}{\mu}}\ ,  
\end{align}
for all $l>0$.

Note that the left-hand side of \eqref{Anderson_type} involves an average over all realisations of the Hamiltonian $H_\om$, which usually corresponds to  different realisations of a random potential. We will refer to \eqref{Anderson_type} as Anderson-type localization.


\section{Logarithmic Lightcone in the LIOM model} \label{sub_L-R}

\begin{Lemma} \label{theo_L-R_LIOM}
Given the LIOM Hamiltonian as defined in \ref{LIOM_Ham}, the dynamics $ U_\om(t)=e^{-it\sum_{r=-L}^L H_{r,\om}} $ gives rise to the following Lieb-Robinson bound. Considering  an operator $ A $, for simplicity supported on the site $ x=0 $, and with $ X=[-l,l] $, it holds:
\begin{align} \label{L-R}
\Eom \, \| \frac{1}{2^{|X^c|}} \left( \Tr_{X^c} U_\om^*(t) A U_\om(t) \right) \otimes \mathds{1}_{X^c} -  U_\om^*(t) A U_\om(t) \| \le  16   \, t \, J \, \xi \, \|A\| \,  e^{  -\frac{l}{2\xi} }
\end{align}
More in general, for a generic support of $ A $, $l$ on the RHS of \eqref{L-R}  is replaced by $ \textrm{dist}(\textrm{supp}(A),X^c) $.
\end{Lemma}

Before the proof we have some remarks.

The fact that the LIOM model gives rise to a logarithmic light cone, as in \eqref{L-R}, implies, because of theorem \ref{Renyi_from_L-R}, that the average dynamical entanglement of any initial product state grows at most as $ \ln \, Jt $.

The factor $ t $ in the RHS of \eqref{L-R} emerges from the application of the Duhamel identities \eqref{Duhamel}, as will become apparent in the proof. The time-dependence of the entanglement entropy in \eqref{EE} again emerges from the application of the Duhamel identities.


Let us consider the limit $ \xi \rightarrow 0 $. This corresponds, according to equation \eqref{decay}, to the set of operators $ \{ H_{r,\om} \} $  being strictly local, meaning that each $ H_{r,\om} $ is supported over a finite region around the point $ r $. This makes the RHS of \eqref{decay} exactly zero. Moreover it is immediate to realize that strictly local $ \{ H_{r,\om} \} $ together with $ [H_{r,\om},H_{s,\om}]=0 $ implies that in the LHS of \eqref{L-R}, whenever $ l $ is large enough such that $ [-l,l] $ contains the supports of the $ H_{r,\om} $'s whose supports contain $ x=0 $, is exactly zero. We call this trivial localization, in fact it is a form of localization without dynamics. In trivial localization there is not a competition among a ``kinetic'' part and a ``potential'' part of the Hamiltonian, to use a language akin to particles' systems. In the work \cite{Farshi_2022} trivial localization refers to the case where in a quantum circuit a 2-qubits gate appears as the product of two single-qubit gate, trivially preventing the spread of the support of an operator.

Let us consider the case $ \xi \ll 1 $. We could think that this would lead to Anderson-type localization. Nevertheless we see that no matter how small is $ \xi $ in \eqref{L-R}, if we take the supremum with respect to time, as we do in the LHS of \eqref{Anderson_type}, then the RHS of \eqref{L-R} immediately saturates to the trivial bound. This means that the theory of localization emerging from the definition of the LIOM Hamiltonian, as far as regards lemma \ref{theo_L-R_LIOM}, does not include, in a meaningful way, Anderson-type localization for many-body systems, according to definition \eqref{Anderson_type}. This inclusion could be achieved, for example, replacing $ t $ with $ t^{\beta} $, with $ \beta \rightarrow 0 $ when $ \xi \rightarrow 0 $, or alternatively when the disorder becomes dominant. 


In the limit $ \xi \gg 1 $ the LIOM model loses its meaning as the description of a local physical system of interacting particles (spins) subjected to disorder, nevertheless becoming an all to all model we expect the Lieb-Robinson bound to become saturated at all times, this is granted by the factor $ \xi $ on the RHS of \eqref{L-R}.

We give a final remark. 
In lemma \ref{theo_L-R_LIOM} we have shown that assuming the existence of LIOM, it follows that the L-R bound of a local Hamiltonian gives rise to a logarithmic lightcone. The converse, namely showing the existence of LIOM assuming a logarithmic lightcone, has not been proven so far. It is immediate to see how a naive approach to the problem fails. Let us consider the mutually commuting and conserved quantities labeled $\{ H_{j,\om} \} $ in equation \eqref{cons_comm}: 
\begin{align} \label{cons_j}
 H_{j,\om} := \lim_{T \rightarrow \infty} \frac{1}{T} \int_0^T dt e^{itH_\om} h_{j,\om} e^{-itH_\om}
\end{align} 
Now, let us assumed that the L-R bound of the system gives a logarithmic lightcone, like in equation \eqref{L-R-beta-res}. We would like to show that, after averaging over the disorder configurations $\om$, the conserved quantities $\{ H_{j,\om} \} $ decay exponentially. This corresponds to show that, with $X$ a connected region containing $j$, the quantity:
 \begin{equation} \label{decay1}
\Eom \, \| \frac{1}{2^{|X^c|}} \left( \Tr_{X^c} H_{j,\om} \right) \otimes \mathds{1}_{X^c} - H_{j,\om} \| 
\end{equation}
decays exponentially.

Let us replace in \eqref{decay1} the definition \eqref{cons_j}, and then use the logarithmic lightcone:
\begin{align}
 & \Eom \, \| \frac{1}{2^{|X^c|}} \left( \Tr_{X^c} H_{j,\om} \right) \otimes \mathds{1}_{X^c} - H_{j,\om} \| \\
 & = \Eom \, \| \lim_{T \rightarrow \infty} \frac{1}{T} \int_0^T dt \left( \frac{1}{2^{|X^c|}} \left( \Tr_{X^c} e^{itH_\om} h_{j,\om} e^{-itH_\om} \right) \otimes \mathds{1}_{X^c} - e^{itH_\om} h_{j,\om} e^{-itH_\om} \right \|  \label{al_last} \\
 & \le  \lim_{T \rightarrow \infty} \frac{1}{T} \int_0^T dt K   \, t^\beta  \, J \,  e^{-\frac{\dist(j,X^c)}{\mu} } \label{last}
\end{align}
We have assumed in \eqref{al_last} that the average can be exchanged with the limit $\lim_{T \rightarrow \infty}$.
The upper bound in \eqref{last} diverges for every $\beta > 0 $. Only for $ \beta = 0 $, that corresponds to what is known as dynamical localization for the Anderson model, \eqref{last} gives a bound exponentially decreasing with distance to \eqref{decay1}. But this shows a fact already know namely that the Anderson model admits a set of (in average) exponentially localized conserved quantities.

Showing the equivalence of the existence of LIOM and of the existence of a dynamics with L-R bounds giving a logarithmic lightcone is still an open problem.

We now proceed with the proof of  \eqref{L-R}.
\begin{proof}
In the following we mostly drop the index $ \om $, denoting the disorder configurations, to light up the notation.
The main idea is to divide the region around $ x=0 $, where $ A $ is supported into an ``inner'' region $ [-d,d] $ with $ d \le l $, and an ``outer'' region $ |r| > d $. $ d $ is a variational parameter that will be used to minimize the corresponding upper bound. At the end of the proof we will find that $ d = \frac{l}{2} $ showing the consistency of our approach. $ H_r $ in the outer region, when $ r > d $ will be replaced with the corresponding restriction in the region $ [1,L] $, and when $ r < -d $ on the region $ [-L,-1] $. These restricted operators will commute with $ A $ because they supports are disjoint. Operators $ H_r $ in the inner region will be replaced by operators supported on $ X = [-l,l] $.

$ [H_r,H_s]=0 $ implies $ U(t):= e^{-it \sum_r H_r } = \prod_r e^{-it H_r } $, and the factors $ e^{-it H_r } $ can be arranged in an arbitrarily order, therefore:
\begin{align}
U(t)^*AU(t)=  \prod_{|r|\le d} e^{it H_r} \prod_{|r|> d} e^{it H_r} A \prod_{|s|> d} e^{-it H_s} \prod_{|s|\le d} e^{-it H_s} 
\end{align}
For the terms $ H_s $ of the outer region we use the Duhamel formula \eqref{Duhamel} to write:
\begin{align} \label{Unitary_Duhamel}
 e^{-it H_s}=e^{-it \Ht_s} - i \int_0^t du e^{-iuH_s}(H_s-\Ht_s)e^{-i(t-u)\Ht_s}
\end{align}
With $ s > d $ we have defined
\begin{align}\label{Hstilde}
 \Ht_s := \frac{1}{2^{|\Lambda \setminus  [1,L]|}} \left( \Tr_{\Lambda \setminus [1,L]} H_s \right) \otimes \mathds{1}_{\Lambda \setminus [1,L]} 
\end{align}
This implies that $ \Ht_s $ is supported on $ [1,L] $, therefore commuting with $ A $. With $ s < -d $ an analogous operator is defined with support on $ [-L,-1] $, commuting with $ A $ as well.

\begin{figure}[h!]
\setlength{\unitlength}{1mm} 

\begin{picture}(145,55)(-30,-50)

\thicklines

\put(-15,-24){\line(0,1){4}}
\put(85,-24){\line(0,1){4}}
\put(-15,-22){\line(1,0){100}}

\put(35,-24){\line(0,1){4}}
\put(65,-24){\line(0,1){4}}
\put(5,-24){\line(0,1){4}}
\put(20,-24){\line(0,1){4}}
\put(50,-24){\line(0,1){4}}

\thinlines

\put(34.25,-28){$ 0 $}
\put(49.25,-28){$ d $}
\put(64.25,-28){$ l $}
\put(84.25,-28){$ L $}
\put(-16.5,-28){$ -L $}

\put(34.25,-31.5){$ \downarrow $}

\put(3.5,-28){$ -l $}
\put(18.5,-28){$ -d $}

\put(28,-36){$ \textrm{supp of} \, A $}

\put(20,-18){$\overbrace{\hspace{30mm}}$}
\put(25,-13){$ \textrm{inner region}$}

\put(5,-9){$\overbrace{\hspace{60mm}}$}
\put(21,-5){$ \textrm{supp of} \, \Hh_s, \, |s| \le d $}

\put(50,-31){$\underbrace{\hspace{35mm}}$}
\put(59,-37){$ \textrm{outer region}$}

\put(36,-41){$\underbrace{\hspace{49mm}}$}
\put(49,-47){$ \textrm{supp of} \, \Ht_s, \, s>d $}

\put(-15,-31){$\underbrace{\hspace{35mm}}$}
\put(-6,-37){$ \textrm{outer region}$}

\put(-16,-41){$\underbrace{\hspace{49mm}}$}
\put(-3,-47){$ \textrm{supp of} \, \Ht_s, \, s<-d $}



\end{picture}
\caption{Schematic representation of the subdivision in regions and of the supports of $\Hh_s$ and $\Ht_s$. } 
\label{supports}
\end{figure}

The second term on the RHS of \eqref{Unitary_Duhamel} is upper bounded in norm by $ t \| H_s - \Ht_s \| $, that, according to the definition \eqref{Hstilde}, and the assumption about the exponential decrease, after disorder averaging \eqref{decay}, of $ H_{r,\om} $, is in turn upper bounded by $ t J e^{-\frac{r}{\xi} } $. Let us now consider, with $ |r|>d $:
\begin{align}
  & e^{it H_r} A  e^{-it H_r}  = e^{it H_r} A \left( e^{-it \Ht_r} - i \int_0^t du e^{-iuH_r}(H_r-\Ht_r)e^{-i(t-u)\Ht_r} \right) \\
  & = e^{it H_r} e^{-it \Ht_r}  A - i e^{it H_r} A  \int_0^t du e^{-iuH_r}(H_r-\Ht_r)e^{-i(t-u)\Ht_r}  \\
  & = \left( e^{it \Ht_r} + i \int_0^t du e^{iuH_r}(H_r-\Ht_r)e^{i(t-u)\Ht_r} \right) e^{-it \Ht_r}  A  - i e^{it H_r} A  \int_0^t du e^{-iuH_r}(H_r-\Ht_r)e^{-i(t-u)\Ht_r}  \\
  & =: A + \delta_r(t) 
\end{align}
Where $ \delta_r(t) $ is an operator supported on the whole lattice $ \Lambda $ and in norm upper bounded by $ 2 t \|A\| J e^{-\frac{r}{\xi} } $. This implies that:
\begin{align}
 e^{it H_{r+1}} e^{it H_r} A  e^{-it H_r} e^{-it H_{r+1}} & =  e^{it H_{r+1}} A e^{-it H_{r+1}} + e^{it H_{r+1}} \delta_{r}(t) e^{-it H_{r+1}} \\ & = A + \delta_{r+1}(t) + e^{it H_{r+1}} \delta_{r}(t) e^{-it H_{r+1}} \label{12}
\end{align}
In \eqref{12} we have that $ \| \delta_{r+1}(t) + e^{it H_{r+1}} \delta_{r}(t) e^{-it H_{r+1}} \| \le 2 t \|A\| J ( e^{-\frac{r}{\xi} } + e^{-\frac{r+1}{\xi} } ) $.
We then have:
\begin{align} \label{def_delta}
  \prod_{|r|> d} e^{it H_r} A \prod_{|s|> d} e^{-it H_s} = A + \delta(t)
\end{align}
with 
\begin{align} \label{final_rest}
 \| \delta(t) \| \le 4 t \|A\| J  e^{-\frac{d+1}{\xi} } \sum_{j=0}^\infty  e^{-\frac{j}{\xi} } = 4 t \|A\| J  e^{-\frac{d}{\xi} } \frac{1}{e^{\frac{1}{\xi}}-1} \le 4 t \xi \|A\| J  e^{-\frac{d}{\xi} }  
\end{align}
We stress that the factor 2 in the upper bound \eqref{final_rest} arises from combining the contributions of the positive and negative $ r $'s in the definition \eqref{def_delta} of $ \delta(t) $.
We remark that the upper bound \eqref{final_rest} is obtained as a consequence of disorder averaging, as stressed after equation \eqref{Hstilde}.

We now consider $ \prod_{|r|\le d} e^{it H_r}  A  \prod_{|s|\le d} e^{-it H_s} $. The idea is to restrict the support of each $ H_r $ of the inner region, $ |r|\le d $, to $ X = [-l,l] $.  We define:
\begin{align}\label{Hhat}
 \Hh_r := \frac{1}{2^{|\Lambda \setminus  [-l,l]|}} \left( \Tr_{\Lambda \setminus [-l,l]} H_r \right) \otimes \mathds{1}_{\Lambda \setminus [-l,l]} 
\end{align}
At this point we proceed similarly to what we have done above for the outer region, we first consider:
\begin{align}
  &  e^{it H_r}  A e^{-it H_r}   = \label{inner_transform} \\
  & = \left( e^{it \Hh_r} + i \int_0^t du e^{iuH_r}(H_r-\Hh_r)e^{i(t-u)\Hh_r} \right) A \left( e^{-it \Hh_r} - i \int_0^t du e^{-iuH_r}(H_r-\Hh_r)e^{-i(t-u)\Hh_r} \right) \\
  & = e^{it \Hh_r} A e^{-it \Hh_r} - i e^{it \Hh_r} A \int_0^t du e^{-iuH_r}(H_r-\Hh_r)e^{-i(t-u)\Hh_r}  + i \int_0^t du e^{iuH_r}(H_r-\Hh_r)e^{i(t-u)\Hh_r} A e^{-it H_r} \\
  & =: e^{it \Hh_r} A e^{-it \Hh_r} + \eta_r(t)
\end{align}
The norm of $ \eta_r(t) $ is upper bounded as: $ \| \eta_r(t) \| \le 2 t\| A \| J e^{-\frac{l+1-r}{\xi} } $. This again follows from disorder averaging, according to the assumption \eqref{decay}.
\begin{align}
 & e^{it H_{r+1}} e^{it H_r} A  e^{-it H_r} e^{-it H_{r+1}}  = \\
 & = \left( e^{it \Hh_{r+1}} + i \int_0^t du e^{iuH_{r+1}}(H_{r+1}-\Hh_{r+1})e^{i(t-u)\Hh_{r+1}} \right) \left(e^{it \Hh_r} A e^{-it \Hh_r} + \eta_r(t) \right) \nonumber \\
 & \hspace{3.5cm} \left( e^{-it \Hh_{r+1}} - i \int_0^t du e^{-iuH_{r+1}}(H_{r+1}-\Hh_{r+1})e^{-i(t-u)\Hh_{r+1}} \right) \\
 & =  e^{it \Hh_{r+1}} \left(e^{it \Hh_r} A e^{-it \Hh_r} + \eta_r(t) \right) e^{-it \Hh_{r+1}} -i e^{it \Hh_{r+1}} e^{it H_r} A e^{-it H_r} \int_0^t du e^{-iuH_{r+1}}(H_{r+1}-\Hh_{r+1})e^{-i(t-u)\Hh_{r+1}} + \nonumber \\ 
 & + i \int_0^t du e^{iuH_{r+1}}(H_{r+1}-\Hh_{r+1})e^{i(t-u)\Hh_{r+1}} e^{it H_r} A 
 e^{-it H_r} e^{-it H_{r+1}} \\
 & =: e^{it \Hh_{r+1}} e^{it \Hh_r} A  e^{-it \Hh_r} e^{-it \Hh_{r+1}} + \eta_{r+1}(t) 
\end{align}
with  
\begin{align}
\| \eta_{r+1}(t) \| \le \|  \eta_{r}(t) \| + 2 t\| A \| J e^{-\frac{l+1-(r+1)}{\xi} } \le 2 t\| A \| J \left( e^{-\frac{l+1-r}{\xi} } + e^{-\frac{l+1-(r+1)}{\xi} } \right)
\end{align}
This implies that:
\begin{align}
 & \prod_{|r|\le d} e^{it H_r}  A  \prod_{|s|\le d} e^{-it H_s} = \\
 & = \prod_{|r|\le d} \left( e^{it \Hh_r} + i \int_0^t du e^{iuH_r}(H_r-\Hh_r)e^{i(t-u)\Hh_r} \right) A \prod_{|s|\le d} \left( e^{-it \Hh_s} - i \int_0^t du e^{-iuH_s}(H_s-\Hh_s)e^{-i(t-u)\Hh_s} \right) \\
 & =: \prod_{|r|\le d} e^{it \Hh_r} A  \prod_{|s|\le d}  e^{-it \Hh_s} + \eta(t)
\end{align}
with
\begin{align} 
 \| \eta(t) \| \le 4 t\| A \| J \sum_{j=0}^{d} e^{-\frac{l+1-j}{\xi} } = 4 t\| A \| J e^{-\frac{l+1}{\xi} } \sum_{j=0}^{d} e^{\frac{j}{\xi} } = 4 t\| A \| J e^{-\frac{l+1}{\xi} } \frac{e^{\frac{d+1}{\xi}}-1}{e^{\frac{1}{\xi}}-1} 
 \le 4 t \xi \| A \| J e^{-\frac{l-d}{\xi} } \label{inner_upper}
\end{align}
In \eqref{inner_upper} we see that with $ \xi \gg d $,  $ \frac{e^{\frac{d+1}{\xi}}-1}{e^{\frac{1}{\xi}}-1} \rightarrow d+1 $, therefore the final bound \eqref{inner_upper} is not good in this limit. Nevertheless the overall bound \eqref{L-R} contains a contribution from $ \|\delta(t)\| $, given in \eqref{final_rest}, that is proportional to $ \xi $, therefore to simplify the form of the final bound we stick to \eqref{inner_upper}.  
We are now ready to prove \eqref{L-R}.
\begin{align}
& \| \frac{1}{2^{|X^c|}} \left( \Tr_{X^c} U^*(t) A U(t) \right) \otimes \mathds{1}_{X^c} -  U^*(t) A U(t) \| = \\ 
& = \| \frac{1}{2^{|X^c|}} \left( \Tr_{X^c} \left( \prod_{|r|\le d} e^{it H_r} ( A + \delta(t) ) \prod_{|s|\le d} e^{-it H_s}   \right) \right) \otimes \mathds{1}_{X^c} -  \left( \prod_{|r|\le d} e^{it H_r} ( A + \delta(t) ) \prod_{|s|\le d} e^{-it H_s}   \right) \| \\
& = \| \prod_{|r|\le d} e^{it \widehat{H}_r} A  \prod_{|s|\le d} e^{-it \widehat{H}_s} + \frac{1}{2^{|X^c|}} \left( \Tr_{X^c} \left( \eta(t) + \prod_{|r|\le d} e^{it H_r} \delta(t) \prod_{|s|\le d} e^{-it H_s}   \right) \right) \otimes \mathds{1}_{X^c} + \\ 
&  \hspace{8cm} - \left( \prod_{|r|\le d} e^{it H_r} ( A + \delta(t) ) \prod_{|s|\le d} e^{-it H_s}   \right) \| \\
& = \|  \frac{1}{2^{|X^c|}} \left( \Tr_{X^c} \left( \eta(t) + \prod_{|r|\le d} e^{it H_r} \delta(t) \prod_{|s|\le d} e^{-it H_s}   \right) \right) \otimes \mathds{1}_{X^c} -  \left( \eta(t) + \prod_{|r|\le d} e^{it H_r}  \delta(t)  \prod_{|s|\le d} e^{-it H_s}   \right) \| \label{33} \\
& \le 2 ( \| \eta(t) \| + \| \delta(t) \| ) \le 8 t \xi \| A \| J \left( e^{-\frac{d}{\xi} }  + e^{-\frac{l-d}{\xi} } \right) \le 16 \, t \, \xi \, \| A \| \, J \, e^{-\frac{l}{2\xi} } \label{34}
\end{align}
In \eqref{33} we have used \eqref{Haar}, see also proposition 1 of \cite{Rastegin_2012}. In \eqref{34} we have minimized the upper bound with respect to $ d $, finding that the minimum is reached for $ d=\frac{l}{2} $. 
\end{proof}



\section{Dynamical generation of entanglement entropy in the LIOM model} \label{EE_from_LIOM}

We want to evaluate $ \Delta \, S(t) := \Eom S\left( \Tr_{[1,L]}  e^{-it\sum_r H_{r,\om}} \rho e^{it\sum_r H_{r,\om}} \right) - S\left( \Tr_{[1,L]} \rho  \right) $, with $ S(\cdot) $ the von Neumann entropy of a density matrix, and $ \rho $  representing the initial state of the system. $ \Delta \, S(t) $ evaluates the generation of  entanglement entropy due to the dynamics of the Hamiltonian \eqref{LIOM}. Numerical and theoretical works  \cite{Znidaric_2008,Bardarson_2012,Serbyn_2013,Kim_2014}, and more recent rigorous work, under the assumption of the LIOM model \cite{Nach_2021,Zeng_2023,Lu_2024}, have shown that the long-time dynamics of the entanglement entropy of MBL systems follows a $ \log t $ law. Within the LIOM model this can be traced back to the slow spread in time of the support of operators as shown in \eqref{L-R}. This is in contrast to the growth in time proportional to $ t $ of the entanglement entropy in generic systems, see \cite{Bravyi_Hastings_Verstraete_2006, Eisert_Osborne_2006, Van_Acoleyen_2013, Marien_2016} and in particular \cite{Toniolo_2024_2} where a linear in time upper bound to the growth of $ \alpha $-Rényi entropies, with $ \alpha \le 1 $, is provided.

It should also be noted that when a large part of the system is traced out, instead of half of it like we are considering here, the  typical behaviour of an initial random pure state under a generic local Hamiltonian dynamics is such that its entanglement entropy grows at very low rates, as proven in \cite{Hutter_2012}.

The following  upper bound \eqref{EE} on $  \Delta \, S(t) $, starting from a generic product state, depends logarithmically on the system's size. As we have seen in theorem \ref{Renyi_from_L-R} of this work, and in other recent results in the literature \cite{Zeng_2023,Lu_2024}, such dependence on the system's size is absent.  In fact lemma \ref{theo_L-R_LIOM} together with theorem \ref{Renyi_from_L-R} implies an upper bound to the dynamically generated von Neumann entropy starting from a product state that for large time grows like $\log t $.  We still report our result because it is obtained in an easier fashion, without recurring to a telescopic sum that is employed by the works that we have mentioned and also by our work \cite{Toniolo_2024_2} upon which theorem \ref{Renyi_from_L-R} is based. The proof is easier but we pay a price. We also stress that the use of telescopic sums for Hamiltonians with exponential decay of interactions, like the LIOM Hamiltonian of equation \eqref{LIOM} is more difficult than in the strictly local case. See for example section V and in particular appendix G of \cite{Toniolo_2024_2}. 
%

\begin{Lemma} \label{lem_EE_from_LIOM}
Given the LIOM Hamiltonian $ H_\om = \sum_{r=-L}^L H_{r,\om} $ defined in \ref{LIOM_Ham} and any initial product state $ \rho=\otimes_{j=-L}^L \rho_j $, with $ \rho_j : \mathds{C}^2 \rightarrow \mathds{C}^2 $, defining $ {\bar t} = 2\xi \ln(L+1)/(J(2\xi \ln(L+1)+1)^2) $,  it holds:
\begin{equation} \label{EE}
\Delta S(t) := \Eom S\left( \Tr_{[1,L]}  e^{-itH_\om} \rho e^{itH_\om} \right) - S\left( \Tr_{[1,L]} \rho  \right) \le
\begin{cases}
  4tJ \xi  + t J (2\xi \ln(L+1)+1)^2 +3  \hspace{5mm} \textrm{with} \hspace{5mm} t \le {\bar t} \\
  2 \xi \ln(2tJ(L+1)) + 2 \xi +3 \hspace{5mm} \textrm{with} \hspace{5mm} t \ge {\bar t}
\end{cases}
\end{equation}
\end{Lemma}

Before the proof we state some remarks.

As we said about the Lieb-Robinson bound \eqref{L-R}, in the limit $ \xi \rightarrow 0 $ the terms $ H_{r,\om} $  of the Hamiltonian becomes, in average, strictly local, together with the assumption that $ [H_{r,\om},H_{s,\om}]=0 $ this implies that the dynamics becomes trivial. Equation \eqref{EE} agrees with this picture, in fact with $ \xi \rightarrow 0 $ we see that $ \bar{t} \rightarrow 0 $, and, with $ t \ge 0 $, $ \xi \rightarrow 0 $ implies $ \Delta S(t) \le 3 $, that is time independent. With $ \xi \rightarrow 0 $ the Hamiltonian becomes (in average) a sum of single-site terms.  The authors of \cite{Nach_2016} found out that for the XY model with random magnetic field the bound on the dynamical generation of entanglement entropy is uniform in time, namely a constant. We stress that in their approach they take the supremum with respect to $ t $ before averaging, that is the correct way of studying dynamical localization in the context of Anderson localized systems. Our approach is different.

The term proportional to $ \ln \, L $ in the upper bound \eqref{EE}, for $ t \ge \bar{t} $,
  as far as regards the system's sizes currently available for experiments or simulations ends up being merely a term of order $ 1 $.

\begin{proof} 
As in the proof of lemma \ref{theo_L-R_LIOM} we drop the disorder index $ \om $, but recall when quantities are given by an average.
The main idea, the same employed in the evaluation of the Lieb-Robinson bound, is to identify an inner region (that is across $ x=0 $), and its complement, called outer region. The factors of $ U(t)=e^{-it\sum_{r=-L}^L H_r} = \prod_{r=-L}^L e^{-it H_r} $ centered in the outer region, meaning with $ r $ in such region, will be approximated with factors that are not supported across the origin, therefore they do not contribute to the entanglement entropy. The factors within the inner region will be approximated with factors strictly supported on a region containing the inner region. Both the approximations, for the factors in the inner and outer region, will lead to density matrices exponentially close in trace norm to the starting ones. To bound the error in the entanglement entropy we will use the Fannes-Audenaert-Petz bound, see \cite{Audenaert_2007} and lemma 1 of \cite{Winter_tight_2016}. The contribution to the entanglement entropy of the factors strictly supported on the inner region will be upper bounded in two different ways corresponding to different time scales. For the short time scale we will employ again the Fannes-Audenaert-Petz bound obtaining an upper bound to the entanglement entropy linear in $t$. For the long time scale we will use the trivial bound to the von Neumann entropy of a density matrix that, in our setting is the number of sites where the density matrix is supported. This will give rise to an upper bound of the type $ \ln J \,t $.

Let us see the details now. We define the inner region to be $ [-l,l] $ (this is a different notation with respect to the proof of \ref{theo_L-R_LIOM}), then define 
\begin{align}
 \rho_l(t):= \prod_{|r|\le l} e^{-itH_r} \rho \prod_{|s|\le l} e^{itH_s}
\end{align}
therefore $ \rho(t) = \prod_{|r| > l} e^{-itH_r} \rho_l(t) \prod_{|s| > l} e^{itH_s} $. As in the proof of the Lieb-Robinson bound we introduce, for $ r > l $, $ \Ht_r $ supported on $[1,L]$, such that:
\begin{align}\label{Htilde}
 \Ht_r := \frac{1}{2^{|\Lambda \setminus  [1,L]|}} \left( \Tr_{\Lambda \setminus [1,L]} H_r \right) \otimes \mathds{1}_{\Lambda \setminus [1,L]} 
\end{align}
Then
\begin{align} 
 e^{-it H_r}=e^{-it \Ht_r} - i \int_0^t du e^{-iuH_r}(H_r-\Ht_r)e^{-i(t-u)\Ht_r} =: e^{-it \Ht_r} + u_r(t) 
\end{align}
with $ \| u_r(t) \| \le tJe^{-\frac{r}{\xi}} $, in the averaged sense from \eqref{decay}.  Analogously we define $ \Ht_r $ for $ r < -l $ supported on $ [-L,-1] $. Then:
\begin{align} \label{product_rest}
 \prod_{|r| > l} e^{-itH_r} =: \prod_{|r| > l} e^{-it\Ht_r} + R(t)
\end{align} 
To evaluate the norm of $ R(t) $ in \eqref{product_rest}, we keep in mind that both sides of \eqref{product_rest} are unitary operators. With $ r > l $, it is
\begin{align}
& e^{-itH_r} e^{-itH_{r+1}} = \left(  e^{-it\Ht_r} + u_r(t) \right) \left(  e^{-it\Ht_{r+1}} + u_{r+1}(t) \right) \\
& = e^{-it\Ht_r}  e^{-it\Ht_{r+1}} + u_r(t)  e^{-it\Ht_{r+1}} + e^{-it H_r} u_{r+1}(t) \label{41}
\end{align}
The norm of the last term in \eqref{41} is upper bounded by $ tJ  e^{-\frac{r+1}{\xi}} $.
It is easy to conclude that, having defined $ R(t) $ in \eqref{product_rest}:
\begin{align} \label{norm_R}
 \|R(t)\| \le 2tJ\sum_{j=l+1}^L e^{-\frac{j}{\xi}} \le 2tJ e^{-\frac{l+1}{\xi}} \sum_{k=0}^\infty e^{-\frac{k}{\xi}} \le 2tJ e^{-\frac{l+1}{\xi}} \frac{1}{1-e^{-\frac{1}{\xi}}} \le 2tJ \xi e^{-\frac{l}{\xi}}
\end{align}
We stress that the factor 2 in the upper bound \eqref{norm_R} arises from bringing together the contributions of the positive and negative $ r $'s in the definition of $ R(t) $.
\begin{align}
 \rho(t) = \left( \prod_{|r| > l} e^{-it\Ht_r} + R(t) \right) \rho_l(t) \left( \prod_{|s| > l} e^{it\Ht_s} + R^*(t) \right) = : \prod_{|r| > l} e^{-it\Ht_r}  \rho_l(t)  \prod_{|s| > l} e^{it\Ht_s} + \widetilde{\rho}(t)
\end{align}
\begin{align}
 \| \widetilde{\rho}(t) \|_1  \le \| \prod_{|r| > l} e^{-it\Ht_r} \rho_l(t) R^*(t) \|_1 + \| R(t) \rho_l(t) \left( \prod_{|s| > l} e^{it\Ht_s} + R^*(t) \right) \|_1 \le 2 \|R(t)\| \| \rho_l(t) \|_1 \le 4tJ \xi e^{-\frac{l}{\xi}}
\end{align}
Denoting $ \eta $ and $ \sigma $ two generic density matrices over $ \mathds{C}^d $, and $ T $ an upper bound on their trace distance, $ \frac{1}{2} \|\eta-\sigma\|_1  \le T \le 1 $, the Fannes-Audenaert-Petz bound \cite{Audenaert_2007,Winter_tight_2016} reads:
\begin{equation} \label{Fannes}
  |S(\eta)-S(\sigma)| \le 
\begin{cases}
T \log_2 (d-1) + H_2(T,1-T) & \textrm{with} \, T \le 1-\frac{1}{d} \\
\log_2 d & \textrm{with} \, T \ge 1 - \frac{1}{d}
\end{cases}
\end{equation}
$ H_2(T,1-T) := -T\log_2 T -(1-T) \log_2 (1-T) $ is the binary Shannon entropy. It is $ H_2(T,1-T)  \le 1 $. Applying the bound \eqref{Fannes}, we get:
\begin{align}
 & S\left( \Tr_{[1,L]}  e^{-itH} \rho e^{itH} \right) = S\left[ \Tr_{[1,L]} \left( \prod_{|r| > l} e^{-it\Ht_r}  \rho_l(t)  \prod_{|s| > l} e^{it\Ht_s} + \widetilde{\rho}(t) \right) \right] - S\left( \Tr_{[1,L]}   \prod_{|r| > l} e^{-it\Ht_r}  \rho_l(t)  \prod_{|s| > l} e^{it\Ht_s}  \right) + \nonumber  \\ & \hspace{7cm} + S\left( \Tr_{[1,L]}   \prod_{|r| > l} e^{-it\Ht_r}  \rho_l(t)  \prod_{|s| > l} e^{it\Ht_s}  \right) \label{46} \\
 & \le 2tJ \xi e^{-\frac{l}{\xi}}(L+1) + 1 +  S\left( \Tr_{[1,L]}   \rho_l(t)  \right) \label{47}
\end{align}
To obtain \eqref{47} we have used, applying \eqref{Fannes}, that $ \| \Tr_X(\eta-\sigma) \|_1 \le \| \eta-\sigma \|_1 $, and the fact that each $ \Ht_r $ is not supported across $ x=0 $ therefore it does affect the von Neumann entropy in \eqref{46}. This last fact is shown, for example, in the equations (6)-(13) of our work \cite{Toniolo_2024_2}. We have also upper bounded the binary Shannon entropy with $ 1 $.
We can anticipate that after the minimization in $ l $ the factor $ L $ in \eqref{47} will be replaced by $ \ln L $.

We recall that we have defined $ \rho_l(t):= \prod_{|r|\le l} e^{-itH_r} \rho \prod_{|s|\le l} e^{itH_s} $, this means that despite the exponential decrease of the tails of the Hamiltonian terms $ H_r $, each $ e^{-itH_r} $ is supported on the all lattice $ \Lambda = [-L,L] \cap \mathds{Z}$. As stated at the beginning of the proof our strategy is to replace each $ H_r $, being $ |r|\le l $, with the truncation of $ H_r $ to the interval $ [-2l,2l] $, denoted $ {\bar H}_r $. This introduces an error in norm at most (for $r=l$) of the order of $ e^{-\frac{l}{\xi}} $. Following closely what done in section \ref{sub_L-R} after \eqref{inner_transform}, we have, replacing in equation \eqref{inner_upper} $ d $ with $ l $, and $ l $ with $ 2l $, that:
\begin{align}
 & \rho_l(t):= \prod_{|r|\le l} e^{-it H_r}  \rho  \prod_{|s|\le l} e^{it H_s} =: \prod_{|r|\le l} e^{-it {\bar H}_r} \rho  \prod_{|s|\le l}  e^{it {\bar H}_s} + {\widetilde{\widetilde{\rho}}}(t)
\end{align}
with
\begin{align}
 \| {\widetilde{\widetilde{\rho}}}(t) \|_1  \le 4 t J \xi  e^{-\frac{l}{\xi} } 
\end{align}
Applying again the Fannes-Audenaert-Petz bound, we have:
\begin{align}
S\left( \Tr_{[1,L]}  e^{-itH} \rho e^{itH} \right) \le 4tJ \xi e^{-\frac{l}{\xi}}(L+1) + 2 + S\left( \Tr_{[1,L]}   \prod_{|r|\le l} e^{-it {\bar H}_r}  \rho  \prod_{|s|\le l} e^{it {\bar H}_s}  \right) 
\end{align}
with $ {\bar H}_r $, supported on $ [-2l,2l] $, defined as in \eqref{Hhat} but with $l$ replaced by $2l$:

\begin{align}\label{Hbar}
 {\bar H}_r := \frac{1}{2^{|\Lambda \setminus  [-2l,2l]|}} \left( \Tr_{\Lambda \setminus [-2l,2l]} H_r \right) \otimes \mathds{1}_{\Lambda \setminus [-2l,2l]} 
\end{align}

We stress that till this point the initial state $ \rho $ is completely generic. We now make the assumption that $ \rho $ is a product state, namely: $ \rho = \otimes_{j=-L}^L \rho_j $. Denoting $ U_{[-2l,2l]}:= \prod_{|r|\le l} e^{-it {\bar H}_r} $, we have that:
\begin{align} \label{product_state}
 & \Tr_{[1,L]} \left( U_{[-2l,2l]} \rho U_{[-2l,2l]}^* \right) = \Tr_{[1,2l]}  \Tr_{[2l+1,L]}   \left( U_{[-2l,2l]} \rho_{[-2l,2l]} U^*_{[-2l,2l]} \otimes \rho_{[-2l,2l]^c} \right) = \nonumber \\ 
 &= \Tr_{[1,2l]} \left( U_{[-2l,2l]} \rho_{[-2l,2l]} U^*_{[-2l,2l]} \right) \otimes \Tr_{[2l+1,L]} \rho_{[-2l,2l]^c} 
\end{align}
It follows that:
\begin{align}
& S\left( \Tr_{[1,L]} U_{[-2l,2l]} \rho U_{[-2l,2l]}^* \right) - S\left( \Tr_{[1,L]} \rho \right) \\
 & = S \left( \Tr_{[1,2l]} \left( U_{[-2l,2l]} \rho_{[-2l,2l]} U_{[-2l,2l]}^* \right) \otimes \Tr_{[2l+1,L]} \rho_{[-2l,2l]^c} \right) - S \left( \Tr_{[1,2l]} \left( \rho_{[-2l,2l]} \right) \otimes \Tr_{[2l+1,L]} \rho_{[-2l,2l]^c} \right) \\
 & = 
  S \left( \Tr_{[1,2l]} \left( U_{[-2l,2l]} \rho_{[-2l,2l]} U_{[-2l,2l]}^* \right)  \right) - S \left( \Tr_{[1,2l]} \left( \rho_{[-2l,2l]} \right) \right) \label{S_prod}
\end{align}
In \eqref{S_prod} to cancel the system-size large contribution coming from $ S \left( \Tr_{[2l+1,L]} \rho_{[-2l,2l]^c} \right) $ we have used $ S(\sigma \otimes \eta) = S(\sigma) + S(\eta) $.
To obtain the final upper bound on $ S\left( \Tr_{[1,L]}  e^{itH} \rho e^{-itH} \right) - S\left( \Tr_{[1,L]} \rho \right) $ we need to upper bound \eqref{S_prod}.
This can be done in two different ways, that will correspond to two different time-scales. For short-time scales we use the Fannes-Audenaert-Petz upper bound, giving:
\begin{align}
 & S \left( \Tr_{[1,2l]} \left( U_{[-2l,2l]} \rho_{[-2l,2l]} U_{[-2l,2l]}^* \right)  \right) - S \left( \Tr_{[1,2l]} \left( \rho_{[-2l,2l]} \right) \right)  \le \frac{1}{2} \|U_{[-2l,2l]} \rho_{[-2l,2l]} U_{[-2l,2l]}^* - \rho_{[-2l,2l]} \|_1 (2l+1) +1 \nonumber \\ 
 & \le t \|\sum_{|r|\le l} {\bar H}_r \| (2l+1) +1 \le t J (2l+1)^2 +1 \label{short_times}
\end{align}
This bound must be compared with the trivial bound for the von Neumann entropy of a state supported on $ 2l+1 $ sites with dimension of the local Hilbert space equal to $ 2 $, that is $ 2l+1 $. This implies that the bound \eqref{short_times} is meaningful for $ t $ such that $ t J (2l+1)^2 +1 \le 2l+1 $, that is: $ t \le 2l/(J(2l+1)^2) $.
Within the short time-scale $ 0 \le t \le  2l/(J(2l+1)^2) $ we perform the minimization with respect to $ l $ of the upper bound
\begin{align} \label{upper_short}
 \Delta S(t) \le 4tJ \xi e^{-\frac{l}{\xi}}(L+1) + t J (2l+1)^2 +3
\end{align}
Since the RHS of the upper bound  \eqref{upper_short} is well defined also with $ l $ real positive, we evaluate the minimum in the standard way for real-valued functions. 
This turns out to be given by the solution of $ e^{\frac{l_{\textrm{min}}}{\xi}}(2l_{\textrm{min}}+1)=L+1 $. It is $ l_{\textrm{min}} < \xi \ln(L+1) $, then since the RHS of \eqref{upper_short} has a unique minimum, replacing $ l= \xi \ln(L+1) $ we get an upper bound to the minimal value of $ \Delta S(t) $. Therefore with $ t \le {\bar t} := 2\xi \ln(L+1)/(J(2\xi \ln(L+1)+1)^2) $ the increase of entanglement entropy is linear in time:
\begin{align} \label{upper_short_final}
 \Delta S(t) \le 4tJ \xi  + t J (2\xi \ln(L+1)+1)^2 +3
\end{align}
With $ t \ge {\bar t} $, we upper bound \eqref{S_prod} with the trivial bound $ 2l +1 $
\begin{align} \label{upper_long}
 \Delta S(t) \le 4tJ \xi e^{-\frac{l}{\xi}}(L+1) + 2l +3
\end{align}
After minimization in $ l $ the minimum $ l_m $ turns out time dependent: $ l_m= \xi \ln (2tJ(L+1)) $, leading to
\begin{align} \label{upper_long_final}
 \Delta S(t) \le 2 \xi \ln(2tJ(L+1)) + 2 \xi +3
\end{align}
This completes the proof.
\end{proof}

\section{Discussion of our results} \label{Discussion}

We discuss and compare our two results given by the bounds in equations \eqref{delta_Renyi_time}, and equation \eqref{EE} for $ t \ge \bar{t} $.

The upper bound \eqref{delta_Renyi_time} is novel because it relies only on the assumption of a logarithmic lightcone \eqref{L-R-beta} without assuming the existence of LIOM (see also the independent work \cite{Zeng_2023}), moreover it is obtained as a result of the general theory developed by us in \cite{Toniolo_2024_2}. Equation  \eqref{delta_Renyi_time} also generalizes the $ \ln \, t$-law to a set of entropies that upper bound the von Neumann entropy, in fact  $ \alpha$-Rényi entropies are decreasing in $\alpha $ and the von Neumann entropy is obtained as the limit $ \alpha \rightarrow 1 $. A shortcoming of our bound \eqref{delta_Renyi_time} is the fact that in the case $ \beta = 0 $ in equation \eqref{L-R-beta}, that corresponds to Anderson-type systems, we are not able to reproduce the results of \cite{Nach_2016} that have found an upper bound that does not depend on $ t $. The dependence of our bound on $ \beta +2 $ can be traced back to the two time-integrations performed in \eqref{def_Delta} and in \eqref{bound_sum_1}.

The upper bound of equation \eqref{EE} suffers from a $\log$-like dependence on the system's size. We think that such dependence could be improved using in the proof of \eqref{EE} a telescopic sum.  We have done it in our work \cite{Toniolo_2024_2}, see section V and in particular appendix G, where we have shown that a generic local Hamiltonian with exponential decay of interactions, leading to a L-R with a linear lightcone, gives dynamical $\alpha$-Rényi entropies (with $0 \le \alpha \le 1$) growing at most linearly in time. But we like to stress that this turns out to be quite difficult.   In Lemma \ref{theo_L-R_LIOM} of section \ref{sub_L-R} we have proven that from the LIOM Hamiltonian, whose terms decay with localization length $ \xi $, follows a logarithm lightcone with localization lenght $ 2\xi $. This means that in the comparison of the results from Lemma \ref{lem_EE_from_LIOM} and \eqref{delta_Renyi_time} we need to identify $ 2\xi $  and $ \mu $, and to consider $ \beta=1 $, $ \alpha=1 $. The scaling with $ \ln (Jt) $ of \eqref{EE}, for $ t\ge \bar{t} $, is $ 2\xi $, namely $ \mu $. This is better than that of \eqref{delta_Renyi_time}, that is $ 3\mu $.

We finally like to stress that Lieb-Robinson bounds, that quantify the dynamical spreading of local operators, may be easier to measure in experiments in comparison to global quantities such as entanglement. As we have proven in theorem \ref{Renyi_from_L-R}, the $ \ln \, t$-law for the entanglement entropy actually follows from Lieb-Robinson bounds with a logarithmic lightcone.

\section*{Acknowledgments}
 Daniele Toniolo and Sougato Bose acknowledge support from UKRI grant EP/R029075/1. D. T. gladly acknowledges discussions about topics related to this work with Abolfazl Bayat, George McArdle, Igor Lerner, Lluís Masanes and Roopayan Ghosh.

\section*{Appendices}
\appendix

\section{Duhamel Identities}

\begin{align}  \label{Duhamel}
 e^{itA}-e^{itB} & = i \int_0^t e^{isA}(A-B)e^{i(t-s)B}ds = i \int_0^t e^{i(t-s)B}(A-B)e^{isA}ds
\end{align}
Proof of the first equality
\begin{align}
 e^{itA}-e^{itB} =  \int_0^t \frac{d}{ds} \left( e^{isA} e^{i(t-s)B} \right) ds = i\int_0^t  e^{isA}(A-B) e^{i(t-s)B}ds 
\end{align}
Proof of the second equality
\begin{align}
 e^{itA}-e^{itB} =  \int_0^t \frac{d}{ds} \left(e^{i(t-s)B} e^{isA} \right) ds = i\int_0^t e^{i(t-s)B} (A-B) e^{isA} ds 
\end{align}


%



%

%


\bibliography{bibliography_MBL_L-R_Bounds}

\end{document}